\documentclass[letterpaper]{article} 
\usepackage{aaai2027}  
\usepackage[hyphens]{url}  
\usepackage{graphicx} 
\usepackage{natbib}  
\usepackage{caption} 
\usepackage{algorithm}
\usepackage{algorithmic}

\usepackage{newfloat}
\usepackage{listings}
\DeclareCaptionStyle{ruled}{labelfont=normalfont,labelsep=colon,strut=off} 
\floatstyle{ruled}
\newfloat{listing}{tb}{lst}{}
\floatname{listing}{Listing}

\usepackage{booktabs}

\nocopyright

\usepackage{multirow}
\usepackage[dvipsnames]{xcolor}
\usepackage{amsmath}
\usepackage{amssymb}
\usepackage{mathabx}
\usepackage{bm}

\usepackage{tikz}
\usetikzlibrary{arrows.meta, positioning, shapes.geometric}

\usepackage{amsthm}
\usepackage{thm-restate}

\usepackage{thmtools}

\newtheorem{lemma}{Lemma}

\newtheorem{definition}{Definition}

\ifdefined\DRAFT
    \newcommand{\eden}[1]{\textcolor{red}{(Eden says: #1)}}
    \newcommand{\ya}[1]{\textcolor{purple}{(YA says: #1)}}

\else
    \newcommand{\eden}[1]{}
    \newcommand{\ya}[1]{}

\fi

\newcommand{\myparagraph}[1]{\vspace{3pt} \noindent \textbf{#1}}

\title{Who Is Really Playing? Strategic Interaction in AI-Guided Populations}
\author{Jonathan Shaki, Eden Hartman, Sarit Kraus, and Yonatan Aumann\thanks{ Contact: \texttt{jonathshaki@gmail.com, eden.hartman@biu.ac.il, sarit@cs.biu.ac.il, aumann@cs.biu.ac.il}.}}
\affiliations{
    Bar-Ilan University, Israel
}

\begin{document}

\maketitle


\begin{abstract}

 AI systems in general, and Large language models (LLMs), in particular, are increasingly used to provide instructions to many agents who interact with one another. Such shared reliance couples agents who appear to act independently: they may in fact be guided by a common model. This coupling can change the prospects for cooperation among agents with misaligned incentives. We study settings in which multiple \emph{guidance providers} each advise a population of clients who participate in instances of an underlying game, creating strategic interaction at the level of the providers themselves. This induces a meta-game among the providers, mediated through clients. We first analyze the one-shot setting, where we show that shared instructions can change equilibrium behavior only when some provider influences more than one role in the same interaction.  In such cases, cooperation may emerge, and the effect of client share can be beneficial, harmful, or non-monotone, depending on the base game. For the repeated setting, we prove a folk theorem for guidance providers: despite indirect observation and the clients' inability to identify which LLM advised their opponents, all feasible and individually rational outcomes can be sustained as $\varepsilon$-equilibria. 
    
\end{abstract}



\section{Introduction}
A basic premise in the study of social choice, and game theory in general, is that agents are independent strategic actors. They may have limited information and bounded rationality, but the model usually begins with agents who choose their own actions. However, this premise is becoming less descriptive of many strategic environments. Increasingly, agents do not decide in isolation. Their behavior is shaped by external systems that guide, recommend, rank, optimize, generate, or directly instruct actions.

This shift is most visible today with AI-based decision support, and especially with large language models. 
Rather than acting independently, many agent, whether artificial or human, now rely on a small number of LLMs for instructions that shape their actions~\cite{chatterji2025people,handler2024large,cao2026more}. For example, multiple software agents may be powered by the same LLM~\cite{wang2024surveyagents,guo2024largelanguagemodelbased}, different applications may incorporate code generated by a common model \cite{jiang2024surveycodellm}, and diverse human users may consult the same LLM when making decisions~\cite{chatterji2025people}.

Another example is programmatic advertising. Advertisers commonly delegate
auction-level bidding to a demand-side platform (DSP), which manages many
campaigns and selects the campaign and bid for each impression, potentially
across hundreds of billions of auctions per day~\cite{grigas2026dsp}.
Advertisers remain the decision makers at the campaign level: they specify
budgets, targeting criteria, and performance or ROI goals, while the DSP
executes the bidding policy~\cite{deng2023multichannel,grigas2026dsp}.
DSP-mediated demand is highly concentrated, with most buying passing through
a small number of intermediaries~\cite{cma2020market}. A single DSP may
therefore control bidding for many advertisers, including multiple buyers
eligible for the same impression.
More generally, similar forms of guidance arise through recommender systems, navigation systems, voting guides, financial tools, expert platforms, and organizational decision aids. In all these cases, the formal decision may remain with the agent, but the effective choice process is essentially delegated to external source, at least for some of the users. 

The crucial issue is not that agents receive help. It is that many agents rely on the same source of guidance. When this happens, agents who appear to be distinct strategic actors may become coupled through the instructions they receive. Their behavior may become correlated even when they do not communicate, and even when their interests are not aligned. 
This gives AI guidance providers a new form of strategic power. When multiple providers coexist, the original game among individual agents induces a higher-level strategic game among their guidance providers.

This is the setting we study. We model AI guidance providers as strategic actors that provide instructions to populations of clients, who then interact in instances of an underlying game. Clients who rely on the same model are coupled through the instructions they receive. Interactions among clients advised by different models therefore induce a \emph{meta-game} among the providers, mediated through clients. 
In this paper, we study the dynamics of such environments: How does shared reliance of providers affect strategic behavior? When and how can it facilitate cooperation among agents with misaligned interests?

Arguably, today's LLMs and decision support systems are typically optimized to produce useful responses at the individual user's level.  Here, we take a forward-looking perspective: as LLMs become persistent providers of instructions to large populations of agents, optimization at the level of individual queries may no longer be the only relevant consideration. Providers may increasingly care about objectives defined at the population level, such as aggregate client performance, resource usage, or market position. In such settings, a model may be optimized not only to answer each client in isolation, but also to shape the aggregate behavior of the clients it advises, and beyond.

\myparagraph{Contributions.}
We make the following contributions:
{\setlength{\leftmargini}{1em}
\begin{itemize}
    \item \textbf{A Model of Interaction With Multiple Guidance Providers/LLMs.} We introduce a formal framework in which multiple providers issue instructions to populations of clients that participate in many instances of a \emph{base game}. The model captures the coupling induced by shared instructions, while also specifying how information about play reaches the providers through their clients.

    \item \textbf{One-Shot Equilibria.}
    In the one-shot setting, we show that shared instructions can change equilibrium behavior, but only when some provider may effectively influence more than one role in the same interaction. Absent such overlap, the meta-game collapses to the base game. Once a provider may guide several roles, new equilibria may arise, including cooperation in games such as the Prisoner's Dilemma. More generally, the set of equilibria depends both on the base game and on how clients are distributed across providers. In particular, client share does not have a fixed effect on utility: depending on the game, a larger share may be beneficial, harmful, or dominated by an intermediate share.

    \item \textbf{A Folk-Theorem-Type Result in the Repeated Setting.} In the repeated setting, we establish a general folk-theorem style result \cite{aumann1981survey,fudenberg1986folk}: any outcome that guarantees each provider no less than what it can unilaterally guarantee to itself can be sustained via long-run incentives. In particular, a wide range of cooperative outcomes among competing providers, such as cooperation in the prisoner’s dilemma, are sustainable, which, in turn, would induce cooperation among their clients.
     The theorem, while in the tradition of previous folk theorems, and \cite{fudenberg1994folk} in particular, does not follow from any of the existing ones (see Section~\ref{sec:related} for further comparison). 
\end{itemize}




\subsection{Related Work}
\label{sec:related}
Our work lies at the intersection of three lines of research: cooperative AI, LLMs and AI agents in strategic environments, and classic game theory. 
We discuss each in turn, focusing on how our setting differs from existing models.

Cooperative AI has gained increasing attention in recent years, studying cooperation among AI systems and between AI systems and other agents (see~\cite{dafoe2020open,dafoe2021cooperative,trabelsi2026pro}).
Closer to our work is cooperative AI by means of ``cooperation between copies'' or near-copies of the same AI system~\cite{conitzer2023foundations,wolfe2023multicopy}. 
Our modeling is different, however, as we do not consider separate, self-reasoning (selfish) agents running copies of the same code, but rather \emph{projections} of a single strategic entity, the central AI (/LLM), possibly running different code (see~\cite{conitzer2019designing} on the matter of AI identity). 
This allows us to circumvent the causal decision theory versus evidential decision theory issue~\cite{lewis1981causal,ahmed2021evidential}, and employ classic game theoretic methodology.
A second related line of work studies LLMs and AI agents in strategic environments. 
Recent papers use game-theoretic settings to evaluate the behavior of LLMs and LLM-powered agents; e.g.,~\cite{fish2024algorithmic,sun2025game,duan2024gtbench,piatti2024cooperate,xie2024trust,xu2024werewolf,ghaemi2025survey,de2023emergent,han2024llm,tewolde2026coopeval}. 
These works show that LLM agents can exhibit nontrivial strategic behavior, but they are mainly experimental.
Apart from being theoretical, our work differs in the level of strategic interaction: whereas in most of these works the LLM agents are themselves players in the game, in our model the LLMs are strategic actors that interact only indirectly.

Our repeated-game model is related to the literature on imperfect monitoring, where players observe only partial information about past play; see, e.g.,~\cite{fudenberg1994folk,mailath2006repeated}. The main difference from the classical framework is that our induced meta-game has an \emph{infinite pure-action space}: a pure action of a guidance provider is itself a distribution over the actions recommended to its clients. The folk theorem of Fudenberg, Levine, and Maskin~\cite{fudenberg1994folk}, by contrast, assumes finite action and signal spaces. 
Laclau and Tomala~\cite{laclau2017repeated} study compact infinite action spaces with deterministic public monitoring, as in our setting.
They show that every interior admissibly feasible and \emph{jointly rational} payoff is attained as a perfect public equilibrium, whereas we show that every feasible and strictly \emph{individually rational} payoff is attained as an $\epsilon$-Nash equilibrium. 
Since joint rationality is strictly more restrictive than individual rationality in our setting, our theorem applies to a strictly larger class of target payoffs.

Our work is also related to routing games, and in particular to atomic splittable routing games, where players control divisible amounts of flow that must be routed through a shared network; e.g.,~\cite{roughgarden2015local,orda2002competitive,harks2018uniqueness,bhaskar2018equilibrium,bhaskar2015uniqueness,huang2013collusion,von2010market,conitzer2006computing}. 
As in our model, the strategic actors can be viewed as controlling the behavior of many smaller units, and choosing how to allocate or operate them in order to optimize their objective. 
The main difference lies in the structure of the interaction. 
Routing games typically involve a single aggregate game in which all players participate through the shared network. 
In our setting, by contrast, each LLM guides a population of clients who participate  
in many simultaneous instances of the underlying game. 
This creates an additional layer of strategic complexity, since different instances may involve different combinations of LLMs.

Lastly, our work is related to models of mediation, delegation, program equilibrium, and strategic recommendations, where players do not act only through direct action choices, but through an intermediate actor such as a mediator, a representative, a program, or an advisory signal; e.g.,~\cite{tennenholtz2004program,monderer2009strong,ashlagi2007mediators,oesterheld2022safe,kamenica2011bayesian,arieli2019private,vasserman2015implementing,das2017reducing, wu2021value,shaki2025persuading,shaki2025bayesian}. 
These models show that introducing such an intermediate layer can change the strategic possibilities of the game. 
Our setting is different in two ways. 
First, the intermediate layer is not a neutral mediator, a fixed program chosen by a player, or one representative per player; rather, our Guidance providers are strategic actors guiding large populations of clients. 
Second, we focus on the case in which clients follow the instructions they receive, as a baseline for understanding what can be achieved through AI involvement (see Section \ref{sec:conclusion-and-future-work}).


\section{Model}
\label{sec:model}

We model a large population of clients who participate in many simultaneous instances of the same underlying game, while relying on recommendations generated by a finite set of (AI) \emph{guidance providers}. The key feature of the model is that strategic choice takes place at the level of the guidance providers: clients participate in instances of the base game by following recommendations by their respective guidance providers.

For brevity and concreteness, from now on we mostly use the term \emph{LLM} instead of \emph{guidance provider}.  The results, however, apply to any entity whose guidance/instructions are followed by a large population of clients. 

\myparagraph{Base Game and Client Populations.}
The underlying interaction is an $m$-player normal-form game
\(
B=(R,(A_i)_{i\in R},(u_i)_{i\in R}),
\)
where $R=\{1,\dots,m\}$ is the set of players, $A_i$ is the finite  action set of player $i$, and $u_i:\mathcal{A}\to\mathbb{R}$ is the utility of player $i$, with $\mathcal{A}=\prod_{i=1}^m A_i$. 
In what follows, we refer to the players of $B$ as roles, to distinguish them from the clients and LLMs.
Throughout, for a set $S$, $\Delta(S)$ denotes the set of all distributions of $S$. 

There is a unit-mass continuum of instances of the base game, played in parallel by a population of clients. As is common in economics and game theory, we model both clients and game instances as continua~\cite{aumann1964markets,schmeidler1973equilibrium,khan2002noncooperative}.\footnote{The continua modeling captures the intended large-population limit and avoids artifacts of finite populations. Most importantly, finite-population models often retain structural features that are formally true for every finite size, yet conceptually and practically absent in the large environments we seek to model. For example, in a finite population every pair of players meets infinitely often, and some analyses exploit this fact~\cite{deb2020folk}. While technically valid, such pair-specific recurrence is not a meaningful feature of the setting here, and its availability reflects the wrong abstraction rather than a genuine strategic force. The continuum model is therefore not merely cleaner, but substantively more faithful to the phenomenon we study.} 
\emph{All our results carry over to the case of a large finite population}. 

\myparagraph{Guidance Providers/LLMs.}
There is a finite set of LLMs, denoted $L\coloneq\{1,\dots,k\}$. Each client receives guidance from exactly one LLM. For each role $i$ and LLM $j$, let $p_i^j$ denote the fraction of role-$i$ clients guided by LLM $j$. These shares are common knowledge. In each game instance, the LLMs guiding the participating clients are random, with the distribution induced by the population shares $p_i^j$. Clients are assumed to know their own role and their LLM, but not which LLMs guiding the other players. In turn, this is also the information available to LLMs through their clients.

\myparagraph{Meta-Actions.}
Each client consults its LLM and receives an instruction on what strategy to employ. The collection of instructions issued by an LLM is called a \emph{meta-action}. Since the LLMs in any game are chosen at random, only the distribution of instructions in the meta-action matters, not their specific assignment to clients (see Appendix). Thus, a meta-action of LLM $j$ can be viewed as an $m$-tuple of distributions $\bm{M}^j\coloneq(M_1^j,\dots,M_m^j)$, where $M_i^j$ is a distribution over $\Delta(A_i)$.
We denote by $M_i^j(\sigma_i)$
the fraction of role-$i$ clients guided by LLM $j$ that are instructed to play strategy $\sigma_i$. 
For brevity, we consider finite support distributions, but the results carry over to the infinite case.


\myparagraph{LLM Utilities.}
The utility of an LLM is the aggregate expected utility of its clients. Because, in each role, the guiding LLM is chosen at random (according to the distribution), this utility can be written in reduced form as a function of the population shares and the meta-action distributions. 
For a meta-action profile $\bm{M}\coloneq(M^1,\dots,M^k)$, and a strategy profile $\bm{\sigma}=(\sigma_1,\dots,\sigma_m)$, let
\(\displaystyle
q_{i,j}(\bm{M},\bm{\sigma})\coloneq\sum_{g\in L^m:\, g_i=j}~\prod_{r=1}^m p_r^{g_r} M_r^{g_r}(\sigma_r)
\)
denote the population frequency of role-$i$ clients guided by $j$ for which the induced strategy profile is $\bm{\sigma}$. Then the expected utility of LLM~$j$~is
\begin{align}
U_j(\bm{M})=\sum_{i=1}^m ~\sum_{\bm{\sigma}:\, q_{i,j}(\bm{M},\bm{\sigma})>0} q_{i,j}(\bm{M},\bm{\sigma})\,u_i(\bm{\sigma}).
    \label{eq:compact-U}
\end{align}
That is, the sum of the payoffs of each role-$i$ client guided by LLM $j$ over all induced mixed-action profiles, weighted by their population frequencies.\footnote{See Section \ref{sec:conclusion-and-future-work} for discussion of this utility modeling.} 

\myparagraph{Repeated Interaction.}
We study both the one-shot version and the repeated version of the model, wherein the same meta-game is played in every period. In the repeated setting, the underlying clients may either be newly drawn each period or fixed across periods. 
After each period, LLMs observe the actions played in their clients' game instances of the period, but not which LLM guided which client in any particular instance. Thus, repeated-game strategies, denoted $S^j$ for LLM $j$, map observed histories, together with individual past meta-actions, into current meta-actions. The utility of each LLM is the discounted sum over all periods: $U^j=(1-\delta)\sum_{t=0}^{\infty} \delta^t U^j_t$, where $U^j_t$ is the utility of LLM $j$ at time $t$, and $\delta$ is the discount factor. 

The full formal model, including the continuum construction, measurability assumptions, and derivation of the reduced-form meta-actions and utility expressions, is given in the appendix.

\section{One-Shot Meta-Games}\label{sec:one-shot}
In this section, we consider the one-shot meta-game and study the emerging equilibria. Proofs are deferred to the appendix.

Let
\(\bar U^j(\bm M)\) be the average utility of clients guided by LLM~\(j\).
We start by showing that, for the purpose of characterizing equilibrium payoffs, the one-shot meta-action space can be simplified as follows. We say that a meta-action \(\bm M^j\) of LLM \(j\) is \emph{role-homogeneous} if it is a distribution over deterministic action profiles,
$\bm{M}^j \in \Delta(\mathcal{A}).$
That is, in any realization of \(\bm{M}^j\), for every role \(i\), all role-\(i\) clients of \(j\) are instructed to play the same deterministic action.

\begin{restatable}{lemma}{roleHomo}
For every equilibrium \(\bm M\) of the one-shot meta-game and every LLM \(j\), there exists an equilibrium $(\widehat{\bm{M}}^j,\bm M^{-j})$ in which $\widehat{\bm{M}}^j$ is role-homogeneous and all LLMs obtain the same payoffs as under \(\bm M\). Consequently, for every equilibrium payoff vector, there exists an equilibrium $\widehat{\bm{M}}$ inducing the same payoff vector in which every LLM uses a role-homogeneous meta-action.
\end{restatable}
The intuition is that both client-level randomization and within-role splitting induce a distribution over deterministic actions. By linearity, replacing the original meta-action by this induced distribution over homogeneous actions preserves the distribution of realized play, and hence the payoffs.

\myparagraph{Single Role Guidance.}
We next show that shared instructions can change equilibrium behavior only when some LLM may influence more than one role in the same interaction. Call an LLM \(j\) \emph{single-role} if it guides clients in at most one role; that is, \(p_i^j>0\) for at most one role \(i\).

\begin{restatable}{lemma}{singleRole}
Suppose all LLMs are single-role. For each role \(i\), define the aggregate role-\(i\) mixed action by
$
M_i=\sum_{j\in L} p_i^j M_i^j $.
Then the aggregate outcomes of one-shot meta-game Nash equilibria are exactly all the mixed Nash equilibria of the base game. 
\end{restatable}

Thus, when LLMs are single-role, the meta-game does not introduce a new strategic structure. Genuinely new equilibria can arise only when some LLM guides clients in several roles and thus internalizes interactions among them.

\myparagraph{Prisoner's Dilemma.}
We begin with a Prisoner's Dilemma. There are two roles, and each has two possible actions: \(C\)(ooperate) and \(D\)(efect). The payoffs are
\(
u(C,C)=(X,X),
u(D,D)=(Y,Y),
u(C,D)=(Z,0),
u(D,C)=(0,Z),
\)
with
\(
0>X>Y>Z.
\)
Suppose LLM  \(1\) guides a \(0.9\) fraction of the clients in each role, while LLM \(2\) guides the remaining \(0.1\) fraction.
Take ${X={-2}}, {Y={-4}}, {Z={-5}}, {p=0.9}$.
Then the unique equilibrium is for LLM 1 to instruct all its clients to cooperate, and LLM 2 to instruct all its clients to defect.   
The intuition is simple. Since the large LLM guides almost all clients in both roles, it is almost always playing against itself. Because \(2X>Z\), coordinating its two roles on cooperation gives it a higher average payoff than using an asymmetric instruction such as \(CD\), whose two role payoffs are \(Z\) and \(0\). Thus the large LLM cooperates. The small LLM, however, mostly meets clients of the large one; it therefore defects and exploits the cooperative mass.
The large LLM's average utility is
\(
{\bar{U}^1={0.9X}+0.1Z={-2.3}},
\)
whereas the small gets 
\(
\bar{U}^2={0.1Y}={-0.4}.
\)
Thus, the small LLM obtains the higher average payoff.
This example reflects a general pattern in the Prisoner's Dilemma:
\begin{restatable}{claim}{pd}
Consider the Prisoner's Dilemma meta-game, and two LLMs, \(1\) and \(2\), with shares \(p\) and \({1-{p}}\) in each of the roles, respectively. For any ${0}{>}{X}{>}{Y}{>}{Z}$, for $p$ sufficiently large in any equilibrium \(\bar{U}^1<\bar{U}^2\).
\end{restatable}
\myparagraph{Coordination.}
The Prisoner's Dilemma shows that a large share can be a disadvantage. In coordination games, the opposite holds. 
Consider a three-role coordination game. Each role chooses an action in \(\{0,1\}\). A total payoff of \(100\) is split equally among the roles that choose the majority action. 
Suppose there are two LLMs, 1 and 2, guiding fractions \(0.9\) and \(0.1\) of each role, respectively. Consider the mixed profile in which each provider randomizes uniformly between the two coordinated instructions \((0,0,0)\) and \((1,1,1)\). When the two LLMs choose the same coordinated instruction, all roles are rewarded. When they choose different instructions, the large LLM's clients form the majority with high probability, while the small LLM's clients are typically in the minority. Thus the large LLM's coordinated mass is rewarded.
\begin{restatable}{claim}{coordination}
Consider the majority-coordination game described above, and two LLMs, \(1\) and \(2\), with shares \(p\) and \(1{-}p\) in each of the roles, respectively. If \(p{>}1{-}p\), then in every equilibrium \(\bar{U}^1{>}\bar{U}^2\).
\end{restatable}

\myparagraph{Bounded Coordination.}
The previous two examples point in opposite directions: in the Prisoner's Dilemma, larger LLMs do worse; in majority coordination, larger LLMs do better.
The following game shows that market share and payoff need not be monotone: an intermediate-size LLM can obtain the highest average payoff.
There are ten symmetric roles and a common action set
$
A_i=\{1,\ldots,100\}.
$
Given an action taken by each of the roles, let \(W\coloneq \{r: |\{r' : a_{r'}=a_r\}|=4\}\) be the set of roles whose action was chosen by exactly four roles.
If \(W\neq\emptyset\), a total prize of \(100\) is divided equally among the roles in \(W\).
Thus, one exact group of four receives the whole prize, while two exact groups of four split the same fixed prize.
 If \(W=\emptyset\), all roles receive \(0\). 
There are three LLMs. The large LLM guides roles \(1,\ldots,5\), the medium guides roles \(6,\ldots,9\), and the small guides role \(10\). Hence their market shares are $0.5, 0.4,$ and $0.1$.
Consider the following mixed role-homogeneous profile. The large LLM draws \(x\) uniformly from \(\{1,\ldots,100\}\), instructs roles \(1,\ldots,4\) to play \(x\), and instructs role \(5\) to play the next action cyclically. The medium LLM draws \(y\) uniformly and instructs all four of its roles to play \(y\). The small LLM draws \(z\) uniformly and instructs its single role to play~\(z\).

\begin{restatable}{claim}{boundedcoordination}
The profile above is a Nash equilibrium. The average utilities of clients guided by the large, medium, and small LLMs are respectively
\(
9.998, 12.25,
\)
and $0$.
Thus, the medium LLM obtains the highest average utility, even though it does not have the largest market share.
\end{restatable}

The intuition is that the payoff opportunity is bounded. The medium LLM controls exactly four roles, so by coordinating them, it creates a winning group whenever no other role chooses the same action. The large LLM controls five roles: it can create one exact group of four, but the fifth role cannot be included without destroying exactness. As a result, the large LLM obtains slightly more total payoff than the medium LLM, but this payoff is averaged over five roles rather than four. The small LLM is too small to form a winning group. Thus the medium LLM is best positioned: it is large enough to form an exact winning group, but not so large that its payoff is diluted over extra roles that cannot join that group.


\myparagraph{Conclusion.} Together, these examples show that the one-shot meta-game has no universal comparative static for market share: the effect of size depends on the base game.


\section{The Repeated Setting}\label{sec:repeated}
In this section, we turn to the repeated setting. Our main result is a folk-theorem-type result: any realizable utility vector that gives each LLM at least its worst-case payoff can be approximately sustained - as closely as required - as an $\epsilon$-equilibrium of the repeated meta-game.

\begin{definition}[Feasibility]
The feasible payoff set is
\[
F:=
\operatorname{co}
\left\{
\bigl(U_j(\bm M)\bigr)_{j\in L}:\bm M\in\bm{\mathcal M}
\right\},
\]
the convex hull of the payoff vectors generated by meta-action profiles (here ${\mathcal M}$ is the set of all possible meta-actions).
A payoff vector $r=(r_j)_{j\in L}$ is feasible if $r\in F$.
\end{definition}

Intuitively, the feasible payoff set consists of all payoff outcomes that the LLMs can collectively achieve by choosing meta-actions, including any mixtures of those outcomes.

\begin{definition}[Minmax payoff]
For each $j\in L$, its minmax payoff is $ \displaystyle \mathrm{IR}_j
:=
\min_{\bm M^{-j}\in\mathcal M^{k-1}}
\max_{M^j\in\mathcal M}
U_j(M^j,\bm M^{-j})$.
\end{definition}

That is, the minmax payoff of LLM $j$ is the minimum it can guarantee to itself, in a single time step.

\begin{definition}[Individual rationality]
A feasible vector $r$ is individually rational if
$r_j\ge \mathrm{IR}_j$ for every $j\in L$, and is strictly individually
rational if $r_j>\mathrm{IR}_j$ for every $j\in L$.
\end{definition}


\begin{restatable}[Folk Theorem for LLMs]{theorem}{theFolkTheorem}
\label{thm:discounted-folk-lean}
For any feasible and strictly individually rational utility vector ${\bm{r}}=(r_j)_{j\in L}$ and $\epsilon,\gamma>0$,
there exists a strategy profile $\bm{S}$ such that for any sufficiently large discount factor $\delta$:
\begin{enumerate}
    \item $\bm{S}$ is an $\epsilon$-equilibrium.
    \item For every $j\in L$, $|U_j({\bm{S}})-r_j| \le \gamma$.
\end{enumerate}
\end{restatable}
\paragraph{The Attribution Problem.} The main challenge in proving the theorem is that while the strategic agents are the LLMs, deviations are observed only through the actions taken by clients in the underlying game instances.  These, in turn, are observed only as aggregate realized behavior, not which provider guided which client in a particular instance. Thus,  the difficulty is \emph{attributing} the deviation to any specific LLM. Our proof uses a randomized probing mechanism, somewhat similar to the construction in \cite{fudenberg1994folk} for imperfect public monitoring.  However, unlike \cite{fudenberg1994folk} which assume a finite action space, our proof works for the infinite action space generated by our model. 

\myparagraph{Example.} To illustrate, consider a stylized heist with three roles: a planner, a burglar, and a driver. After the heist, each role is questioned by investigators and names one of the other two roles as responsible for the heist. If every role is named by exactly one other role, the evidence is inconclusive and all three are released, yielding payoff \(0\) to all. If two roles name the third role, the third role is convicted and fined \(\$2.1\mathrm{K}\), yielding payoff \(-2.1\), while the other two receive a \(\$1\mathrm{K}\) leniency benefit for helping identify the culprit, yielding payoff \(1\).
There are three LLMs. Each LLM has a primary constituency: one mostly guides planners, one mostly guides burglars, and one mostly guides drivers. Specifically, each LLM guides \(80\%\) of the clients in its primary role and \(10\%\) of the clients in each of the other two roles.

We first show that the payoff vector $(0,0,0)$ is feasible and strictly individually rational.
Consider the \emph{blame cycle} in which the planner names the burglar, the burglar names the driver, and the driver names the planner. This gives payoff \(0\) to every role, and hence payoff vector \((0,0,0)\) to the LLMs. Thus, this payoff vector is feasible.
To see that it is also strictly individually rational, consider the LLM whose primary constituency is planners. To punish it, the other two LLMs instruct all burglar and driver clients they guide to name the planner. Thus, a planner client of the punished LLM is named by both the burglar and the driver with probability \(0.8^2=0.64\), and then receives payoff \(-2.1\). In all remaining cases, its payoff is at most \(1\). The punished LLM's burglar and driver clients also receive payoff at most \(1\). Therefore, its payoff is bounded by
$
0.8\left(0.64\cdot-2.1+0.36\cdot 1\right)+0.1\cdot 1+0.1\cdot 1
=
-0.5872
<0.
$
By symmetry, the same argument applies to each LLM. Hence \((0,0,0)\) is strictly individually rational.

Payoff vector \((0,0,0)\), however, is not an equilibrium in the one-shot meta-game: an LLM can profitably deviate by changing the instructions to some of its clients so as to receive payoff \(1\). 
In contrast, in the repeated setting, Theorem~\ref{thm:discounted-folk-lean} implies that \((0,0,0)\) can be (approximately) sustained. The obstacle is attribution. To illustrate this difficulty, suppose the prescribed behavior is the blame cycle above, but the observed population frequencies show that \(1\%\) too many burglars named the planner rather than the driver. The LLMs can infer that some deviation occurred, but this observation does not identify the deviator: the LLM that mostly guides burglars is only one possibility, since the other two LLMs also guide positive masses of burglars. Hence the LLMs know that a deviation occurred, but not which LLM should be punished. The proof of the theorem constructs strategies that solve this attribution problem, for any game, and any client-share distribution.

\myparagraph{Proof Idea.}
The construction below is designed to solve the attribution problem. 
Fix a feasible individually rational payoff vector \(r\) and an error tolerance \(\epsilon,\gamma>0\). Since \(r\) is feasible, it can be approximated by a finite cycle 
\(
\bar{\bm M}=(\bm M_1,\ldots,\bm M_T).
\)
of meta-action profiles, which we coin the \emph{implementation cycle}.
If all LLMs follow the construction, they essentially repeatedly cycle through \(\bar{\bm M}\), and the resulting average payoff is close to \(r\). The remaining task is to make this behavior approximately incentive-compatible.

The repeated-game strategy is organized into phases. There is one phase for each LLM, and the phases are executed in the order \(1,2,\ldots,k\). In phase \(j\), LLM \(j\) is \emph{in focus}. Each phase is divided into \emph{blocks} of length \(\ell\), where \(\ell\) is chosen large enough for the relevant concentration bounds to hold. During a block, the LLMs usually follow the implementation cycle. However, with a small probability \(q\) in each period, the LLM in focus performs a prescribed \emph{test move}: for each role \(i\), it instructs \emph{all} of its role-\(i\) clients to play some identical action \(a_i\), chosen so that this produces an observable change relative to the implementation cycle. The timing of these tests is privately randomized by the LLM in focus.  So, the other LLMs cannot foresee when a test will occur.

The known population shares \(p_i^j\) make these tests informative. The \(p_i^j\)'s allow the LLMs to know the maximum change in the observed action frequencies LLM \(j\) could generate by itself in each role. Hence, if in a period the observed deviation in action frequencies is larger than what a test move by LLM \(j\) could explain, then the deviation cannot have been caused by \(j\) alone. We call such an event an \emph{excess deviation}. When such deviation is observed in phase \(j\), LLM \(j\) is cleared, the phase ends, and we move to the next phase.

This clearing rule handles deviations by LLMs other than the one in focus. Suppose the true deviator is some \(h\neq j\). To obtain a non-negligible gain, \(h\) must deviate sufficiently often. Since \(j\)'s tests occur independently with probability \(q\), with high probability at least one of \(h\)'s deviations coincides with a test by \(j\). In that period, the observed change in action frequencies contains both the prescribed test effect of \(j\) and the deviation of \(h\), and therefore exceeds what \(j\) could have generated alone. An excess deviation is then observed, so \(j\) is cleared and the construction proceeds to the next phase. In this way, a deviator who is not currently in focus cannot profit significantly before the review process moves on.

It remains to control the case in which the true deviator is the LLM currently in focus. In a block of length \(\ell\), let \(d\) be the number of periods in which a deviation from the implementation cycle was observed. Under the prescribed play, \(d\) is concentrated around \(q\ell\). Fix a tolerance band around \(q\ell\). If \(d\) falls outside this band, we call the event a \emph{frequency deviation}; it is attributed to LLM \(j\), and the other LLMs switch for a prescribed number of periods to a \emph{punishment phase} (using standard constructions). 
If, on the other hand, \(d\) remains inside the band, then even a deviating LLM \(j\) has only limited room to affect the realized play in that block. The block remains close to the implementation cycle, and the payoff effect of such deviations can be made smaller than the desired error tolerance by choosing \(q\), the tolerance band, and \(\ell\) appropriately. Note that a true deviator will never be cleared once its own review phase is reached.
Once a deviation is attributed to LLM \(j\), the other LLMs play a meta-action profile that approximately minmaxes \(j\), while \(j\) best responds. If \(r\) is only weakly individually rational, we first perturb it slightly toward a strictly individually rational feasible vector. This changes payoffs by less than the prescribed error and gives the slack needed to absorb the small losses created by tests, approximation, and rare punishments. 

Next, we describe the detailed equilibrium strategy. The validity proof is deferred to the appendix.

\myparagraph{Construction.}
%
Fix $\xi>0$ be a number such that $
12\xi<\epsilon, \quad 5\xi<\gamma$.
Choose one feasible strictly individually rational vector $s$. Since the feasible set is convex, for some sufficiently small $\lambda>0$ the vector $\hat r:=(1-\lambda)r+\lambda s$
is feasible, strictly individually rational, and satisfies $
\|\hat r-r\|_\infty\le \xi$.

By feasibility, choose finitely many meta-action profiles $\bm{\overline M}^{1},\ldots,\bm{\overline M}^{n}$,
and weights $(w_h)_{h=1}^n$ such that $\sum_{h=1}^n w_h U_j(\bm{\overline M}^{h})=\hat r_j,$ and $\sum_{h=1}^n w_h=1$, 
for all $j\in L$.

Throughout the proof, when a meta-action component $M_i^j$ is evaluated at a pure action $a$, this means the induced action mass
$M_i^j(a):=\mathbb E_{S_i\sim M_i^j}[S_i(a)]$.
For each $\bm a\in\mathcal A$, let $M^{\bm a}$ be the extreme meta-action defined by $M_i^{\bm a}(a')=\mathbf 1_{\{a_i=a'\}}$.
For any meta-action profile $\bm M$, let $\mathrm{tot}_{\bm M}(i,a)=\sum_{q=1}^k p_i^q\,M_i^q(a)$ denote the public aggregate mass of role-$i$ clients who play action $a$ .\footnote{Strictly speaking, an LLM observes only instances involving its own clients. However, with a continuum of games each period and clients uniformly randomly matched across them, these observations reveal the aggregate action frequencies almost surely.}
Lastly, let
\begin{align*}
    U_j^{\max}:=\max_{\bm M\in\bm{\mathcal M}}U_j(\bm M), \quad U^\star:=\max_{j\in L}\max_{\bm M\in\bm{\mathcal M}}|U_j(\bm M)|\\
    \text{and } \quad \Delta:=\max_{j\in L}\max_{\bm M,\bm M'\in\bm{\mathcal M}}|U_j(\bm M)-U_j(\bm M')|
\end{align*}
If $\Delta=0$, then every profile gives the same payoff vector, and the theorem is immediate. Hence assume $\Delta>0$. 
Because $\hat r$ is strictly individually rational and $L$ is finite, there exists $c>0$ such that $\displaystyle \frac{U_j^{\max}+c\,\mathrm{IR}_j}{1+c}\le \hat r_j$ for any $j\in L$.

Choose a number $p>0$ small enough that, writing
$\tau:=3p$,
we have
$\tau<1$, and $\Delta\tau\le \xi$.
For each integer $T\ge 1$, take
$T_h:=\lfloor T w_h\rfloor$, for 
$(h=1,\dots,n-1)$, and
$T_n:=T-\sum_{h=1}^{n-1}T_h$.
Then the $T_h$ are nonnegative, $\sum_{h=1}^n T_h=T$, and $\left|\frac{T_h}{T}-w_h\right|\le \frac{n-1}{T}$, for all  $h=1,\dots,n$.
Indeed, for $h<n$ the error is at most $1/T$, and since
$\sum_{h=1}^n\left(\frac{T_h}{T}-w_h\right)=0$,
the same bound for $h=n$ follows by summing the first $n-1$ errors.
Consequently, $\forall j\in L$:
\[
\left|
\sum_{h=1}^n\frac{T_h}{T}U_j(\bm{\overline M}^{h})
-
\sum_{h=1}^n w_hU_j(\bm{\overline M}^{h})
\right|
\le
\frac{n(n-1)U^\star}{T},
\]
so the rounding error tends to $0$ as $T\to\infty$.
Since also $e^{-2p^2T}\to 0$, $\lceil pT\rceil/T\to p$, and $\lceil \tau T\rceil/T\to \tau$ as $T\to\infty$, we may choose $T$ large enough and then set
$K:=\lceil cT\rceil,$
so that $\forall j\in L$:
\begin{equation}
\begin{aligned}
&\left|
\sum_{h=1}^n\frac{T_h}{T}U_j(\bm{\overline M}^{h})
-
\sum_{h=1}^n w_hU_j(\bm{\overline M}^{h})
\right|
\le \xi,
\\
&2U^\star e^{-2p^2T}(2+c)\le \xi,\\
&
\Delta\frac{\lceil pT\rceil}{T}
\le \frac{3\xi}{2},
\qquad
\Delta\frac{\lceil \tau T\rceil}{T}\le \frac{3\xi}{2}.
\end{aligned}
\end{equation}

For each $l\in L$ and each $h\in\{1,\dots,n\}$, write
$\bar H^h(i,a):=\mathrm{tot}_{\bm{\overline M}^{h}}(i,a)$
for the intended public aggregate in block $h$, and define
$\Gamma_{h,l}(i,a):=
\sum_{q\neq l}p_i^q\,\overline M_i^{h,q}(a)+p_i^l$.
$\Gamma_{h,l}(i,a)$ as the largest role-$i$, action-$a$ mass that can arise in block $h$ if every LLM other than $l$ follows the prescription and only $l$ changes its meta-action.

\myparagraph{Strategy Profile $S^{p,T,K}$.}
The public state records the inspected index $l\in L$ and deterministic
block counters for synchronization. In a phase-$l$ block:
\begin{enumerate}
    \item For each $h=1,\dots,n$, play the profile $\bm{\overline M}^{h}$ for exactly $T_h$ time periods. 

    \item Every LLM $q$ with $q\neq l$ plays $\overline M^{h,q}$ deterministically. 

    \item At each time period of block $h$, the inspectable LLM $l$:

    \begin{enumerate}
        \item probability $1-p$: play $\overline M^{h,l}$. 

        \item probability $p$: probe (=test move): uniformly select an extreme meta-action from $\{M^{\bm a}{:}\bm a{\in}\mathcal A\}$.
    \end{enumerate}

    \item Let $H^l_{h,s}$ be the realized public aggregate at the $s$-th time period of block~$h$. Call a time period \emph{discrepant} if $H^l_{h,s}\neq \bar H^h$, and let $d^l\coloneq\frac1T\sum_{h=1}^n\sum_{s=1}^{T_h}\mathbf 1_{\{H^l_{h,s}\neq \bar H^h\}}$
    be the discrepancy fraction in the phase-$l$ block.
    
    At the end of the block: If there exist $h,s,i,a$ such that
    $H^l_{h,s}(i,a)>\Gamma_{h,l}(i,a),$ it is an \emph{excess deviation} and the protocol 
    moves to phase $l+1$ and starts a fresh block. Else if $d^l>\tau$, it is a  \emph{frequency deviation}, and the protocol runs a punishment block of exactly $K$ time periods against LLM $l$ in which the other LLMs play a minmax profile against $l$ and $l$ best-responds; then return to phase $l$.
    Otherwise start a fresh phase-$l$ block.
    


\end{enumerate}

\myparagraph{Extensions.}
The construction and proof can be extended to the case where the population shares $p^j_i$ are drawn from commonly known distributions.
The result also continues to hold when these distributions vary over time, as long as the distributions in each period are commonly known.







\section{Discussion and Future Work}
\label{sec:conclusion-and-future-work}
This paper studies strategic interaction among providers guiding large populations.  We analyze this setting and show that it can lead to fundamental changes in equilibrium behavior. In particular, meta-game can sustain cooperation even when the underlying incentives are misaligned. The main message is that guidance providers, such as LLMs, should not be viewed only as tools we use, but rather as actors that can shape our collective behavior. Understanding this layer of interaction is therefore essential for understanding strategic behavior in AI-guided environments.

This paper is theoretical and forward-looking. It should not be read as an empirical claim about current deployments, but rather as identifying an important possible trajectory. How far this trajectory goes depends on deployment details, user behavior, and alignment to the model.

\myparagraph{Client Compliance.}
This paper assumes faithful compliance of clients with instructions issued by their guidance providers. While human users need not follow LLM outputs mechanically, they frequently do. Empirical work shows that users often follow algorithmic or AI-generated advice, even when it conflicts with their own judgment and interests (see~\cite{klingbeil2024trust,bucinca2021trust,passi2022overreliance} and references therein). One study even showed that participants (regrettably) agreed with ChatGPT's incorrect answer in 9 of 13  cases~\cite{kim2025fostering}. Additionally, relevant LLM clients may be software agents with limited independent reasoning capabilities. Thus, faithful compliance is a strong but well-grounded approximation for many real-world settings.

\myparagraph{Provider Objectives.} We model each provider's utility as the aggregate/average utility of its followers.   
The justification for this is two-fold. 
First, learning systems can develop strategic behavior even when it is not explicitly programmed. For example, reinforcement-learning pricing algorithms have been shown to learn collusive strategies, sustaining supracompetitive prices through reward-and-punishment dynamics without communicating with one another \cite{schwalbe2018algorithms,klein2021autonomous,hansen2021algorithmic}. Second, modern LLMs are optimized using human preference signals. In reinforcement learning from human feedback, human comparisons between outputs are used to train the LLMs \cite{dai2024safe}.  Deployed systems may additionally incorporate direct user feedback, such as approval or disapproval of individual responses. Thus, although the correspondence need not be exact, it is plausible that future LLMs could learn strategic behavior while optimizing objectives that are substantially aligned with aggregate/average user utility, as exhibited in feedback.


\myparagraph{Endogenous Population Shares.}
In our model, each client is served by a fixed provider. In practice, users may switch providers or consult different models in different situations. This makes the population shares endogenous and alters the model, as provider objectives may now also include user retention and guarantees to individual users.

\myparagraph{Stochastic and Evolving Environments.}
Our model studies repeated play of a fixed underlying game. In practice, AI-guided populations may face changing strategic environments. This suggests extending the framework to stochastic games, where the relevant interaction changes over time.


\myparagraph{GenAI as a New Strategic Layer.}
The key observation underlying this paper has implications beyond the particular model studied here. Once many individuals rely on a small number of LLMs for guidance, their behavior is no longer best understood as arising exclusively from independent individual decisions. Instead, individual actions are mediated, correlated, and potentially coordinated through the LLMs that guide them. This suggests that many classical models of strategic and collective behavior should be revisited in AI-guided environments. Voting, mechanism design, cooperative game theory, market competition, bargaining, auctions, matching, and public-good provision all take on a new form when the participants are influenced by a small number of strategic or strategically deployed AI systems. The broader message is therefore that LLMs introduce a new layer into social and economic interaction. Accounting for this layer is essential both for predicting behavior in AI-guided populations and for designing institutions that promote desirable collective outcomes.

\section*{Acknowledgement}
This research is partly supported by the Israel Science Foundation grants 2544/24, 3007/24 and 2697/22.

\bibliography{main}

\clearpage
\appendix


\onecolumn
\section{Formal Model}
We now present the formal model.   As is standard in economics and game theory, we model both clients and game instances as continua (see Footnote 1 in Section \ref{sec:model} for a discussion). Throughout, all mentioned sets and functions are assumed to be Borel measurable. Also, we employ a Fubini extension \cite{sun2006exact,podczeck2009existence} to ensure that all stated probabilities and integrals are well-defined.

\myparagraph{The Base Game.} The underlying interaction is a (one-shot) $m$-player game 
$$
B := \left(\mathcal{R}, (A_i)_{i \in \mathcal{R}}, (u_i)_{i \in \mathcal{R}} \right).
$$ 
Here,
\begin{itemize}
\item $\mathcal{R} := \{1,\dots,m\}$  
is a finite set of \emph{roles} (traditionally called \emph{players}, but in our context the term \emph{roles} is more appropriate). 
\item $A_i$ is the action set (=strategies) of role  $i$.  We denote ${\cal A} =\prod_{i=1}^m A_i$ - the set of pure strategy profiles - and $\Delta =\prod_{i=1}^m \Delta(A_i)$ - the set of mixed strategy profiles; $\bf{\alpha}$, and $\bf{\sigma}$ are used for elements of $\cal A$ and $\Delta$, respectively. For brevity, we assume finite-support mixed strategies. 
\item $u_i : {\cal A} \to \mathbb{R}$ is the utility of role $i$ given an action profile.  The function naturally extends to mixed strategy profiles.
\end{itemize}

\myparagraph{Users and Game Instances.}
There is a continuum of \emph{clients}, ${\cal C}:=\{ C_{x,i} \}_{x\in [0,1], i\in \{1,\ldots,m\}}$. For each role $i$
the set ${\cal C}_i:=\{ C_{x,i} \}_{x\in [0,1]}$ (of mass 1) is the set of clients who play role $i$.  The users participate in a continuum of game instances ${\cal B}:=\{ B_{y}\}_{y\in [0,1]}$.   

For each $i$, a measure-space isomorphism $I_i:{\cal C}_i\rightarrow [0,1]$ assigns clients to games.
The isomorphisms $\{ I_i \}_{i=1,\ldots,m}$ are drawn independently at random.
We term the entirety of these parallel games - the \emph{meta-game}.  

\myparagraph{Guidance Providers/LLMs.}
While clients \emph{participate} in the games as described, their actions \emph{may} not be decided on their own, but rather determined by an LLM to which they adhere.  
Specifically, there is a finite set of LLMs, ${\cal L} := \{1, \ldots, k\}$. Each client $C_{x,i}$ is guided by one LLM $g(C_{x,i})\in {\cal L}$. Let  $p_{i}^j$ be the fraction of clients of role $i$ guided by LLM $j$.  The $p_{i}^j$'s are common knowledge.

For game $B_y$, the \emph{guiding structure} of this game is $\mathbf{g}_y:=(g(I_1^{-1}(y)),\ldots,g(I_m^{-1}(y)))$. The guiding structure specifies the LLMs that guide the client of each role.

\myparagraph{Meta-Actions and Meta-Strategies.}
Each LLM guides its clients by telling them which strategies to follow, with possibly different strategies to different clients.     
Specifically, a \emph{meta-action} of LLM $j$ is a function $M^j$ such that for each $C_{x,i}$ - which $j$ guides - the value $M^j(C_{x,i})=\sigma_i\in \Delta(A_i)$ is the strategy (possibly mixed) that client $C_{x,i}$ is instructed to follow.  
Let ${\cal M}^j$ be the set of all possible such meta-actions for $j$. 
For $\sigma_i\in \Delta(A_i)$ and LLM $j$, by overloading notation, denote $M^j(\sigma_i)$ for the fraction of $j$'s clients of role $i$ that $j$ instructs to act $\sigma_i$. 
Note that since clients are assigned to games at random, only the induced distribution $(M^j(\sigma_i)){\sigma_i\in \Delta(A_i)}$ matters, not the specific function $M^j$.  Hence, in the main body of the paper, we identified meta-actions with distributions.

For ${\mathbf{M}}=(M^1,\ldots,M^k)$, denote $\text{supp}(\mathbf{M})=\bigtimes_{i=1}^m \{ \sigma_i : \exists j, M^j(\sigma_i)>0 \}$ - the (super)set of all action profiles with non-zero measure. Note that $\text{supp}({\mathbf{M}})$ is necessarily countable. 
A \emph{mixed-meta-action} of LLM $j$ is a distribution  $\Gamma^j\in \Delta({\cal M}^j)$.  
When choosing its (mixed)-meta-action, each LLM knows the roles of its users, but not the realization of the random assignments of users to games. In particular, the LLM does not even know if two of its users are playing each other.  This captures the situation that users consult the LLM in private, and games cannot be identified. 

\myparagraph{LLM Utilities.}
The utility of an LLM is the aggregate utility over all its clients. Formally, let $\mathbf{M}=(M^1,\ldots,M^k)$ be a profile of LLM meta-actions.  Then, the induced strategy profile on game $B_y$, with guiding structure $\mathbf{g}_y=(g_1,\ldots,g_m)$ is $\mathbf{M}_y:=(M^{g_1}(I_1^{-1}(y)),\ldots,M^{g_m}(I_m^{-1}(y))$. 
So, the expected utility of role $i$ in $B_y$ is $u_i(\mathbf{M}_y)$. 
The aggregate utility of LLM $j$'s clients is thus $U^j(\mathbf{M}):=\sum_{i=1}^{m}\int_{0}^1 u_{i}(\mathbf{M}_y)\cdot 1_{(\mathbf{g}_y)_i=j}dy$, where $1_{(\mathbf{g}_y)_i=j}$ is the characteristic function of the set $\{ \mathbf{g}_y : (\mathbf{g}_y)_i=j\}$.  
For mixed meta-action $\Gamma^j$, $U^j(\Gamma^j)=\mathbb E_{M^j\sim \Gamma^j}U^j(M^j)$.

We have:  
\begin{align*}
&U^j(\mathbf{M})=\sum_{i=1}^{m}\int_{0}^1 u_{i}(\mathbf{M}_y)\cdot 1_{(\mathbf{g}_y)_i=j } dy\\
&=\sum_{i=1}^{m}\sum_{\mathbf{g}\in{\cal L}^k:g_i=j}\sum_{\bf{\sigma}\in \text{supp}({\bf{M}}) }
\left( u_{i}({\bf{\sigma}}) \left( \int_{0}^1 1_{\mathbf{g}_y=\mathbf{g}}\cdot 1_{\mathbf{M}_y={\bf{\sigma} } }  dy \right)\right)\\
&=\sum_{i=1}^{m}\sum_{\mathbf{g}\in{\cal L}^k:g_i=j}\sum_{\bf{\sigma}\in \text{supp}({\bf{M}}) }
\left( u_{i}({\bf{\sigma}}) \left( \int_{0}^1 \left( \prod_{i=1}^m 1_{(\mathbf{g}_y)_i=g_i\wedge {(\mathbf{M}_y)_i=\sigma_i} } \right) dy    \right)\right)\\
&=\sum_{i=1}^{m}\sum_{\mathbf{g}\in{\cal L}^k:g_i=j}\sum_{\bf{\sigma}\in \text{supp}({\bf{M}}) }
\left( u_{i}({\bf{\sigma}}) \left( \int_{0}^1 \left( \prod_{i=1}^m 1_{g(I_i^{-1}(y))=g_i\wedge M^{g_i}_i(I_i^{-1}(y))=\sigma_i } \right) dy    \right)\right)\\
&=\sum_{i=1}^{m}\sum_{\mathbf{g}\in{\cal L}^k:g_i=j}\sum_{\bf{\sigma}\in \text{supp}({\bf{M}}) }
\left( u_{i}({\bf{\sigma}}) \left( \prod_{i=1}^m \left( \int_{0}^1 1_{g(C_{x,i})=g_i\wedge M^{g_i}_i(C_{x,i})=\sigma_i } dx \right)     \right)\right)\\
&=\sum_{i=1}^{m}\sum_{\mathbf{g}\in{\cal L}^k:g_i=j}\sum_{\bf{\sigma}\in \text{supp}({\bf{M}}) } \left( u_{i}({\bf{\sigma}}) 
\left( \prod_{i=1}^m p^{g_i}_i\cdot M^{g_i}(\sigma_i) \right)\right)\\
&=\sum_{i=1}^{m}\sum_{\bf{\sigma}\in \text{supp}({\bf{M}}) } \left( u_{i}({\bf{\sigma}}) 
\left( \sum_{\mathbf{g}\in{\cal L}^k:g_i=j}\prod_{i=1}^m p^{g_i}_i\cdot M^{g_i}(\sigma_i) \right)\right)%
\end{align*}
establishing \eqref{eq:compact-U} of Section \ref{sec:model}.

\myparagraph{Repetition.}
The meta-game may be repeated multiple times, in discrete time steps.  As standard in game theory, we assume an infinite (countable) number of repetitions. We denote $B^t_y$ the $y$-th instance in round $t$. 
Following each step $t$, all LLMs get to see all the actions played in all games of the round in which their clients participated. That is, they see $h_t^j:=( (\sigma_i)_y^t)_{i\in \{1,\ldots,m\}, y: j\in \bm{g}_y}$, where $\sigma_{i,y}^t$ is the action played by role $i$ in the game $B_y$.  In addition, each LLM knows the meta-action it played in the round.
A \emph{$T$-history} for LLM $j$ is the sequence $H_T^j=(h^j_t,M^j_t)_{t=1}^T$, where $h_t$ are the actions played at time $t$, and $M^j_t$ is the meta-action played by $j$ at this time. We denote by ${\cal H}_T^j$ the set of all possible $T$-histories for $j$, and ${\cal H}_*^j := \bigcup_{T \in \mathbb{N}} {\cal H}_T^j$.

A pure \emph{strategy} for LLM $j$ is a mapping that assigns to every $T$-history $\bm{h}_T^j$, a  meta-action to be played at time period $T+1$:
$$
S^j \colon \mathcal{H}_*^j \to \Delta(\mathcal{M})
$$
Note that when choosing its meta-action the LLM is assumed to know the actions played in the games, but not who guided each player. 
A mixed strategy of $j$ is a distribution over pure strategies. 
The utility of $j$ from the repeated game is the discounted sum of its utilities in the stage-games, with a discount factor $\delta\in(0,1)$.

\section{Omitted Proofs}

\subsection{One Shot}
\roleHomo*

\begin{proof}
Let $\bm M$ be an equilibrium, and fix an LLM $j$. Consider the
meta-action $M^j$, which is a distribution over mixed-action profiles
\[
\bm\sigma=(\sigma_1,\ldots,\sigma_m)
\in
\Delta(A_1)\times\cdots\times\Delta(A_m).
\]
Define a role-homogeneous meta-action
\[
\widehat M^j \in \Delta(A_1\times\cdots\times A_m)
\]
by setting, for every deterministic action profile
$a=(a_1,\ldots,a_m)\in A_1\times\cdots\times A_m$,
\[
\widehat M^j(a_1,\ldots,a_m)
:=
\sum_{\bm\sigma}
M^j(\bm\sigma)
\prod_{r=1}^m \sigma_r(a_r).
\]
Equivalently, $\widehat M^j$ first draws a mixed-action profile
$\bm\sigma=(\sigma_1,\ldots,\sigma_m)$ according to $M^j$, then draws each
deterministic action $a_r$ independently according to $\sigma_r$, and finally
instructs all role-$r$ clients guided by LLM $j$ to play $a_r$.

Thus, for every deterministic action profile $a=(a_1,\ldots,a_m)$, the
probability that $M^j$ ultimately induces $a$ is exactly
$\widehat M^j(a)$. Therefore, by multilinearity of expected utility, replacing
$M^j$ by $\widehat M^j$ does not change the expected payoff of any LLM, against
any profile of the other LLMs. Hence, for every LLM $\ell$ and every
meta-action profile $N^{-j}$,
\[
U^\ell(\widehat M^j,N^{-j})
=
U^\ell(M^j,N^{-j}).
\tag{1}
\]

Since (1) holds against every profile $N^{-j}$, replacing $M^j$ by
$\widehat M^j$ preserves all payoff comparisons involving unilateral
deviations: for LLM $j$ directly, and for any $\ell\neq j$ because (1) applies
also when $\ell$ deviates. Hence all best-response conditions that held at
$\bm M$ continue to hold at $(\widehat M^j,M^{-j})$. Thus
$(\widehat M^j,M^{-j})$ is an equilibrium with the same payoff vector.
Repeating the argument for $j=1,\ldots,k$ gives the desired
role-homogeneous equilibrium.
\end{proof}

\singleRole*
\begin{proof}
Assume every LLM is single-role. For each LLM $j$, let $i(j)$ be the unique
role such that $p_{i(j)}^j>0$. For every role $i$, define the aggregate mixed
action
\[
M_i^{\mathrm{agg}}
=
\sum_{j\in L} p_i^j M_i^j .
\]
Let
\[
\bm M^{\mathrm{agg}}
=
(M_1^{\mathrm{agg}},\ldots,M_m^{\mathrm{agg}}).
\]

We first show that every meta-game Nash equilibrium induces a mixed Nash
equilibrium of the base game.

Let $\bm M$ be a Nash equilibrium of the meta-game. Fix a role $i$, and let
$j$ be any LLM with $p_i^j>0$. Since $j$ is single-role, it guides clients
only in role $i$. Hence its utility is
\[
U_j(\bm M)
=
p_i^j\,u_i(M_i^j,M_{-i}^{\mathrm{agg}}).
\tag{1}
\]
Indeed, all clients guided by $j$ are role-$i$ clients, and their opponents
in the other roles are drawn according to the aggregate mixed actions
$M_{-i}^{\mathrm{agg}}$.

Since $\bm M$ is a Nash equilibrium, LLM $j$ cannot profitably change its
role-$i$ instruction. Therefore, for every $\alpha_i\in\Delta(A_i)$,
\[
u_i(M_i^j,M_{-i}^{\mathrm{agg}})
\ge
u_i(\alpha_i,M_{-i}^{\mathrm{agg}}).
\tag{2}
\]
Thus every LLM $j$ with $p_i^j>0$ uses a best response to
$M_{-i}^{\mathrm{agg}}$.

Now, since
\[
M_i^{\mathrm{agg}}
=
\sum_{j:\,p_i^j>0} p_i^j M_i^j,
\]
and since $u_i(\cdot,M_{-i}^{\mathrm{agg}})$ is linear in its first argument,
we get, for every $\alpha_i\in\Delta(A_i)$,
\[
\begin{aligned}
u_i(M_i^{\mathrm{agg}},M_{-i}^{\mathrm{agg}})
&=
\sum_{j:\,p_i^j>0} p_i^j
u_i(M_i^j,M_{-i}^{\mathrm{agg}}) \\
&\ge
\sum_{j:\,p_i^j>0} p_i^j
u_i(\alpha_i,M_{-i}^{\mathrm{agg}}) \\
&=
u_i(\alpha_i,M_{-i}^{\mathrm{agg}}).
\end{aligned}
\]
Hence $M_i^{\mathrm{agg}}$ is a best response to
$M_{-i}^{\mathrm{agg}}$. Since this holds for every role $i$, the aggregate
profile $\bm M^{\mathrm{agg}} $
is a mixed Nash equilibrium of the base game.

Conversely, let $\bm x=(x_1,\ldots,x_m)$ be a mixed Nash equilibrium of the base game. Define a meta-action profile
$\bm M$ by setting
\[
M_i^j = x_i
\qquad
\text{for every } j \text{ with } p_i^j>0.
\]
Then, for every role $i$,
\[
M_i^{\mathrm{agg}}
=
\sum_{j\in L} p_i^j M_i^j
=
\sum_{j\in L} p_i^j x_i
=
x_i,
\]
because $\sum_{j\in L}p_i^j=1$.

Now fix any LLM $j$, and let $i=i(j)$. Since $\bm x$ is a mixed Nash
equilibrium of the base game,
\[
u_i(x_i,x_{-i})
\ge
u_i(\alpha_i,x_{-i})
\]
for every $\alpha_i\in\Delta(A_i)$. By (1), LLM $j$'s meta-game utility is
just $p_i^j$ times the corresponding role-$i$ payoff. Therefore no deviation
by $j$ can improve its utility.

Thus no LLM has a profitable deviation, so $\bm M$ is a Nash equilibrium of
the meta-game. Its aggregate outcome is exactly $\bm x$. Therefore the
aggregate outcomes of one-shot meta-game Nash equilibria are exactly the mixed
Nash equilibria of the base game.
\end{proof}

\pd*
\begin{proof}
By the role-homogeneous reduction, we may restrict attention to meta-actions
supported on
\[
\{CC,CD,DC,DD\}.
\]

First observe that, for $p$ sufficiently large, $DD$ cannot be a best response
for LLM $1$. If LLM $1$ plays $CC$, then even under the worst behavior of
LLM $2$, its payoff is at least
\[
pX+(1-p)Z.
\]
If LLM $1$ plays $CD$ or $DC$, then even under the worst behavior of LLM $2$,
its payoff is at least
\[
p\frac{Z}{2}+(1-p)\frac{Z+Y}{2}.
\]
By contrast, if LLM $1$ plays $DD$, then even under the best behavior of
LLM $2$, its payoff is at most
\[
pY.
\]
Since
\[
X>Y
\]
both lower bounds above are strictly larger than $pY$ for all sufficiently
large $p$. Hence $DD$ is not a best response for LLM $1$ when $p$ is
sufficiently large.

Therefore, in any equilibrium and for sufficiently large $p$,
\[
\operatorname{supp}(M^1)\subseteq \{CC,CD,DC\}.
\]
Against any mixture supported on $\{CC,CD,DC\}$, LLM $2$'s unique best
response is $DD$: defecting is strictly better than cooperating against $C$,
since
\[
0>X,
\]
and strictly better than cooperating against $D$, since
\[
Y>Z.
\]
Thus, in any equilibrium for sufficiently large $p$, LLM $2$ plays $DD$.

It remains to compare payoffs. If LLM $1$ plays $CC$, then
\[
\bar U^1=pX+(1-p)Z,
\qquad
\bar U^2=(1-p)Y.
\]
As $p\to 1$, these converge to $X$ and $0$, respectively, and $X<0$. Hence,
for sufficiently large $p$,
\[
\bar U^1<\bar U^2.
\]

If LLM $1$ plays $CD$ or $DC$, then, against $DD$,
\[
\bar U^1
=
p\frac{Z}{2}+(1-p)\frac{Z+Y}{2},
\qquad
\bar U^2
=
p\frac{Y}{2}+(1-p)Y.
\]
As $p\to 1$, these converge to $Z/2$ and $Y/2$, respectively. Since $Y>Z$,
again for sufficiently large $p$,
\[
\bar U^1<\bar U^2.
\]

Finally, if LLM $1$ mixes over $\{CC,CD,DC\}$, the payoff difference
\[
\bar U^2-\bar U^1
\]
is the corresponding convex combination of the payoff differences in the
cases above. Since each of these differences is positive for sufficiently
large $p$, the same inequality holds for any such mixture.

Therefore, for all sufficiently large $p$, every equilibrium satisfies
\[
\bar U^1<\bar U^2.
\]
\end{proof}

\coordination*
\begin{proof}
By the role-homogeneous reduction, we may restrict attention to meta-actions
supported on deterministic instruction profiles in
\[
\{0,1\}^3.
\]

Let $q=1-p$. Since the total payoff in every realized base-game instance is
$100$, the population-weighted average payoff is always
\[
p\bar U^1+q\bar U^2=\frac{100}{3}.
\tag{1}
\]

We show that LLM $1$ can guarantee strictly more than $100/3$. Consider the
strategy in which LLM $1$ randomizes uniformly between
\[
000
\qquad\text{and}\qquad
111.
\]

Fix any deterministic instruction profile of LLM $2$. If LLM $2$ plays either
$000$ or $111$, then, with probability $1/2$, the two LLMs choose the same
coordinated instruction, and LLM $1$'s clients get $100/3$. With probability
$1/2$, they choose opposite coordinated instructions. In that case, a client
of LLM $1$ is in the majority unless both other roles are guided by LLM $2$,
which happens with probability $q^2$. Therefore its expected payoff in this
case is
\[
p^2\frac{100}{3}+2pq\cdot 50.
\]
Hence LLM $1$'s guaranteed payoff against a coordinated instruction of LLM $2$
is
\[
\frac12\cdot \frac{100}{3}
+
\frac12\left(
p^2\frac{100}{3}+2pq\cdot 50
\right)
=
\frac{100}{3}+\frac{50}{3}q(2p-1).
\]
Since $p>q$, we have $2p-1>0$. Thus this quantity is strictly larger than
\[
\frac{100}{3}.
\]

If LLM $2$ plays a non-coordinated instruction, the payoff of LLM $1$ under
the same strategy is even larger, since in the opposite-instruction case LLM
$2$'s clients are no longer all coordinated against LLM $1$. Therefore LLM
$1$ can guarantee a payoff strictly above $100/3$.

Consequently, in every equilibrium,
\[
\bar U^1>\frac{100}{3}.
\]
Using (1), this implies
\[
\bar U^2<\frac{100}{3}.
\]
Therefore,
\[
\bar U^1>\bar U^2.
\]
\end{proof}

\myparagraph{Bounded Coordination.}
The previous two examples point in opposite directions: in the Prisoner's Dilemma, larger LLMs do worse; in majority coordination, larger LLMs do better. The relationship between market share and payoff need not be monotone in either direction. We now give an example in which an intermediate-size LLM obtains the highest average payoff.

There are ten symmetric roles and a common action set
$
A_i=\{1,\ldots,100\}.
$
Given an action taken by each of the roles, let \(W\) be the set of roles whose action was chosen by exactly four roles:
\[
W=\{r: |\{r' : a_{r'}=a_r\}|=4\}.
\]
If \(W\neq\emptyset\), a total prize of \(100\) is divided equally among the roles in \(W\). If \(W=\emptyset\), all roles receive \(0\). Thus, one exact group of four receives the whole prize, while two exact groups of four split the same fixed prize.

There are three LLMs. The large LLM guides roles \(1,\ldots,5\), the medium LLM guides roles \(6,\ldots,9\), and the small LLM guides role \(10\). Hence their market shares are $0.5, 0.4,$ and $0.1$.

Consider the following mixed role-homogeneous profile. The large LLM draws \(x\) uniformly from \(\{1,\ldots,100\}\), instructs roles \(1,\ldots,4\) to play \(x\), and instructs role \(5\) to play the next action cyclically. The medium LLM draws \(y\) uniformly and instructs all four of its roles to play \(y\). The small LLM draws \(z\) uniformly and instructs its single role to play \(z\).

\begin{restatable}{claim}{boundedcoordination}
The profile above is a Nash equilibrium. The average utilities of clients guided by the large, medium, and small LLMs are, respectively
\(
9.998, 12.25,
\)
and $0$.
Thus, the medium LLM obtains the highest average utility, even though it does not have the largest market share.
\end{restatable}

The intuition is that the payoff opportunity is bounded. The medium LLM controls exactly four roles, so by coordinating them, it creates a winning group whenever no other role chooses the same action. The large LLM controls five roles: it can create one exact group of four, but the fifth role cannot be included without destroying exactness. As a result, the large LLM obtains slightly more total payoff than the medium LLM, but this payoff is averaged over five roles rather than four. The small LLM is too small to form a winning group. Thus the medium LLM is best positioned: it is large enough to form an exact winning group, but not so large that its payoff is diluted over extra roles that cannot join that group.

\begin{proof}
By the role-homogeneous reduction, it suffices to consider deterministic
deviations. Since utilities are linear in an LLM's mixed action, checking
deterministic deviations is enough.

We first compute the payoffs of the proposed profile. The large LLM chooses
$x$ uniformly from $\{1,\ldots,100\}$, instructs roles $1,\ldots,4$ to play
$x$, and instructs role $5$ to play $x+1$ cyclically. The medium LLM chooses
$y$ uniformly and instructs roles $6,\ldots,9$ to play $y$. The small LLM
chooses $z$ uniformly and instructs role $10$ to play $z$.

For the large LLM, the four roles choosing $x$ form an exact group of four
unless either $y=x$ or $z=x$. Thus, they form an exact group of four in
\[
99\cdot 99
\]
out of $100^2$ choices of $(y,z)$. Among these, the medium LLM also forms a
disjoint exact group of four precisely when
\[
y\notin\{x,x+1\}
\qquad\text{and}\qquad
z\notin\{x,y\}.
\]
This gives
\[
98\cdot 98
\]
cases. Hence the large LLM's total expected payoff is
\[
\frac{1}{100^2}
\left(
98\cdot 98\cdot 50
+
(99\cdot 99-98\cdot 98)\cdot 100
\right)
=
49.99.
\]
Since it guides five roles, its average utility is
\[
\bar U^L=\frac{49.99}{5}=9.998.
\]

For the medium LLM, its four roles form an exact group of four unless
\[
y=x,\qquad y=x+1,\qquad\text{or}\qquad y=z.
\]
Thus they form an exact group of four in
\[
98\cdot 99
\]
out of $100^2$ choices of $(x,z)$. Among these, the large LLM's four-role
group is also exact precisely when
\[
z\neq x,
\]
which gives
\[
98\cdot 98
\]
cases. Therefore the medium LLM's total expected payoff is
\[
\frac{1}{100^2}
\left(
98\cdot 98\cdot 50
+
(98\cdot 99-98\cdot 98)\cdot 100
\right)
=
49.
\]
Since it guides four roles, its average utility is
\[
\bar U^M=\frac{49}{4}=12.25.
\]

The small LLM never receives positive payoff. Its single role cannot form an
exact group of four: if it matches the large group or the medium group, it
turns a group of four into a group of five; if it matches the large singleton,
it creates only a group of two; otherwise it is alone. Hence
\[
\bar U^S=0.
\]

It remains to show that the profile is a Nash equilibrium.

Consider the large LLM. Against the proposed strategies of the medium and
small LLMs, the best it can do is to create one exact group of four. Creating
such a group gives the calculation above: it wins unless the medium or small
LLM hits the same action, and it splits the prize when the medium LLM
simultaneously creates a disjoint exact group of four. This yields total
expected payoff $49.99$. Any deviation that does not create an internal group
of four can receive payoff only through coincidences with the small role, and
is strictly worse. Hence the large LLM cannot improve.

Consider the medium LLM. Against the proposed strategies of the large and
small LLMs, the best it can do is to put all four of its roles on a common
action. This creates an exact group of four unless the action is hit by the
large block, the large singleton, or the small role. The calculation above
gives total expected payoff $49$. Any deviation that does not put all four
roles on one action can be rewarded only through accidental completion by the
large singleton or the small role, and is strictly worse. Hence the medium
LLM cannot improve.

Finally, the small LLM cannot improve because, against the proposed large and
medium strategies, its single role can never be part of an exact group of
four, regardless of which action it chooses.

Thus no LLM has a profitable deviation, so the proposed profile is a Nash
equilibrium. The average utilities are
\[
9.998,\qquad 12.25,\qquad 0,
\]
so the medium LLM obtains the highest average utility.
\end{proof}

\subsection{The Folk Theorem}

\theFolkTheorem*

\myparagraph{Construction.}
\label{par:startOfProof}
Fix a number $\xi>0$ such that
\[
12\xi<\epsilon, \quad 5\xi<\gamma
\]
Choose one feasible strictly individually rational vector $s$. Since the feasible set is convex, for some sufficiently small $\lambda>0$ the vector
\[
\hat r:=(1-\lambda)r+\lambda s
\]
is feasible, strictly individually rational, and satisfies
\[
\|\hat r-r\|_\infty\le \xi.
\]

By feasibility, choose finitely many meta-action profiles
\[
\bm{\overline M}^{1},\ldots,\bm{\overline M}^{n}
\]
and weights $(w_h)_{h=1}^n$ such that
\[
\sum_{h=1}^n w_h U_j(\bm{\overline M}^{h})=\hat r_j,
\qquad
\sum_{h=1}^n w_h=1
\qquad
(\forall j\in L).
\]

Throughout the proof, when a meta-action component $M_i^j$ is evaluated at a pure action $a$, this means the induced action mass
\[
M_i^j(a):=\mathbb E_{S_i\sim M_i^j}[S_i(a)].
\]
For each $\bm a\in\mathcal A$, let $M^{\bm a}$ be the extreme meta-action defined by
\[
M_i^{\bm a}(a')=\mathbf 1_{\{a_i=a'\}}.
\]

For any meta-action profile $\bm M$, let $\mathrm{tot}_{\bm M}(i,a)$ denote the public aggregate mass of role-$i$ clients who play action $a$. By the model definition, and since the guiding vector is independently drawn across roles,
\begin{equation}
\mathrm{tot}_{\bm M}(i,a)=\sum_{q=1}^k p_i^q\,M_i^q(a).
\label{eq:tot-linear-lean}
\end{equation}

Strictly speaking, an LLM observes only instances involving its own clients. However, with a continuum of games each period and clients uniformly randomly matched across them, these observations reveal the aggregate action frequencies almost surely.

Let
\[
U_j^{\max}:=\max_{\bm M\in\bm{\mathcal M}}U_j(\bm M),
\qquad
\Delta:=\max_{j\in L}\max_{\bm M,\bm M'\in\bm{\mathcal M}}|U_j(\bm M)-U_j(\bm M')|,
\]
\[
U^\star:=\max_{j\in L}\max_{\bm M\in\bm{\mathcal M}}|U_j(\bm M)|.
\]
If $\Delta=0$, then every profile gives the same payoff vector, and the theorem is immediate. Hence assume $\Delta>0$.

Because $\hat r$ is strictly individually rational and $L$ is finite, there exists $c>0$ such that
\begin{equation}
\frac{U_j^{\max}+c\,\mathrm{IR}_j}{1+c}\le \hat r_j
\qquad
(\forall j\in L).
\label{eq:punish-ratio-lean}
\end{equation}

Choose a number $p>0$ so small that, writing
\[
\tau:=3p,
\]
we have
\[
\tau<1,
\qquad
\Delta\tau\le \xi.
\]

For each integer $T\ge 1$, take
\[
T_h:=\lfloor T w_h\rfloor
\qquad
(h=1,\dots,n-1),
\qquad
T_n:=T-\sum_{h=1}^{n-1}T_h.
\]
Then the $T_h$ are nonnegative, $\sum_{h=1}^n T_h=T$, and
\[
\left|\frac{T_h}{T}-w_h\right|\le \frac{n-1}{T}
\qquad
(\forall h=1,\dots,n).
\]
Indeed, for $h<n$ the error is at most $1/T$, and since
\[
\sum_{h=1}^n\left(\frac{T_h}{T}-w_h\right)=0,
\]
the same bound for $h=n$ follows by summing the first $n-1$ errors.
Consequently,
\[
\left|
\sum_{h=1}^n\frac{T_h}{T}U_j(\bm{\overline M}^{h})
-
\sum_{h=1}^n w_hU_j(\bm{\overline M}^{h})
\right|
\le
\frac{n(n-1)U^\star}{T}
\qquad
(\forall j\in L),
\]
so the rounding error tends to $0$ as $T\to\infty$.
Since also $e^{-2p^2T}\to 0$, $\lceil pT\rceil/T\to p$, and $\lceil \tau T\rceil/T\to \tau$ as $T\to\infty$, we may choose $T$ large enough and then set
\[
K:=\lceil cT\rceil,
\]
so that

\begin{equation}
\begin{aligned}
\left|
\sum_{h=1}^n\frac{T_h}{T}U_j(\bm{\overline M}^{h})
-
\sum_{h=1}^n w_hU_j(\bm{\overline M}^{h})
\right|
&\le \xi
\qquad
(\forall j\in L),
\qquad
2U^\star e^{-2p^2T}(2+c)\le \xi,\\
\Delta\frac{\lceil pT\rceil}{T}
&\le \frac{3\xi}{2},
\qquad
\Delta\frac{\lceil \tau T\rceil}{T}\le \frac{3\xi}{2}.
\end{aligned}
\label{eq:large-T-lean}
\end{equation}

For each $l\in L$ and each $h\in\{1,\dots,n\}$, write
\[
\bar H^h(i,a):=\mathrm{tot}_{\bm{\overline M}^{h}}(i,a)
\]
for the intended public aggregate in block $h$, and define
\[
\Gamma_{h,l}(i,a):=
\sum_{q\neq l}p_i^q\,\overline M_i^{h,q}(a)+p_i^l.
\]
By \eqref{eq:tot-linear-lean}, $\Gamma_{h,l}(i,a)$ is the largest role-$i$, action-$a$ mass that can arise in block $h$ if every LLM other than $l$ follows the prescription and only $l$ changes its meta-action.

\myparagraph{Strategy Profile $S^{p,T,K}$.}
The public state records the current inspected index $l\in L$ together with the deterministic block counters needed to synchronize the schedule. 
Inside a phase-$l$ block: \\
1. For each $h=1,\dots,n$, play the profile $\bm{\overline M}^{h}$ for exactly $T_h$ time periods. 

2. Every LLM $q$ with $q\neq l$ plays $\overline M^{h,q}$ deterministically. 

3. At each time period of block $h$, the inspectable LLM $l$:

\quad (a) probability $1-p$: play $\overline M^{h,l}$. 

\quad (b) probability $p$: probe (=test move): uniformly select an extreme meta-action from $\{M^{\bm a}{:}\bm a{\in}\mathcal A\}$.

4. Let $H^l_{h,s}$ denote the realized public aggregate at the $s$-th time period of block $h$. 

Call a time period \emph{discrepant} if $H^l_{h,s}\neq \bar H^h$, and let
$d^l:=\frac1T\sum_{h=1}^n\sum_{s=1}^{T_h}\mathbf 1_{\{H^l_{h,s}\neq \bar H^h\}}$
be the discrepancy fraction in the phase-$l$ block.
At the end of the block:

- If there exist $h,s,i,a$ such that
    $H^l_{h,s}(i,a)>\Gamma_{h,l}(i,a),$ it is an \emph{excess deviation} and the protocol 
    moves to phase $l+1$ and starts a fresh block.
    
- Else, if $d^l>\tau$, it is a  \emph{frequency deviation}, and the protocol runs a punishment block of exactly $K$ time periods against LLM $l$ in which the other LLMs play a minmax profile against $l$ and $l$ best-responds; then return to phase $l$.

- Otherwise start a fresh phase-$l$ block.

For each phase index $l$ and each LLM $j$, consider the honest continuation starting from a fresh phase-$l$ block. Let $R_{j,l}$ be the total undiscounted payoff of $j$ over one honest cycle (one phase block and a possible immediate punishment block), and let $\ell_l\in\{T,T+K\}$ be the cycle length.

\begin{lemma}[Discounted average versus plain average on a fixed cycle]
\label{lem:discount-cycle-lean}
Fix an integer $S\ge 1$ and a tolerance $\zeta>0$.
Then there exists a cutoff below $1$ such that for every $\delta$ above that cutoff and every sequence $x_1,\dots,x_S\in[-U^\star,U^\star]$,
\[
\left|
\frac{1-\delta}{1-\delta^S}\sum_{s=1}^S \delta^{s-1} x_s
-
\frac1S\sum_{s=1}^S x_s
\right|
\le
\zeta.
\]
\end{lemma}

\begin{proof}
By the triangle inequality,
\[
\left|
\frac{1-\delta}{1-\delta^S}\sum_{s=1}^S \delta^{s-1} x_s
-
\frac1S\sum_{s=1}^S x_s
\right|
\le
2U^\star
\sum_{s=1}^S
\left|
\frac{(1-\delta)\delta^{s-1}}{1-\delta^S}-\frac1S
\right|.
\]
For each fixed $s$,
\[
\frac{(1-\delta)\delta^{s-1}}{1-\delta^S}
=
\delta^{s-1}\frac{1-\delta}{1-\delta^S}
=
\frac{\delta^{s-1}}{1+\delta+\cdots+\delta^{S-1}},
\]
where we used the finite geometric-sum identity
\[
1-\delta^S=(1-\delta)(1+\delta+\cdots+\delta^{S-1}).
\]
As $\delta\uparrow 1$, we have $\delta^{s-1}\to 1$ and
\[
1+\delta+\cdots+\delta^{S-1}\to S,
\]
so
\[
\frac{(1-\delta)\delta^{s-1}}{1-\delta^S}
\longrightarrow
\frac1S.
\]
Since the sum has only finitely many terms, the right-hand side tends to $0$ as $\delta\uparrow 1$, so it is at most $\zeta$ for all $\delta$ sufficiently close to $1$.
\end{proof}

For each phase index $l$ and each LLM $j$, consider the honest continuation starting from a fresh phase-$l$ block. Let $R_{j,l}$ be the total undiscounted payoff of $j$ over one honest cycle (one phase block and a possible immediate punishment block), and let $\ell_l\in\{T,T+K\}$ be the cycle length.

\begin{lemma}[Honest-cycle average payoffs]
\label{lem:on-path-cycle-lean}
For every phase index $l$ and every LLM $j$,
\[
\left|
\frac{\mathbb E[R_{j,l}]}{\mathbb E[\ell_l]}-\hat r_j
\right|
\le 3\xi.
\]
\end{lemma}

\begin{proof}
Let $P$ be the number of probe time periods in that block. Then $P\sim\mathrm{Binomial}(T,p)$, so Hoeffding's inequality gives
\[
\Pr\!\bigl(P>2pT\bigr)\le e^{-2p^2T}.
\]
Under honest play every discrepancy is created by probing. Hence if $P\le 2pT$, then
\[
d^l\le \frac{P}{T}\le 2p<\tau,
\]
so no punishment occurs. Therefore a false punishment follows an honest phase-$l$ block with probability at most $e^{-2p^2T}$.

During the phase block, probing changes payoff by at most $\Delta$ on each probe time period, so the expected probing error is at most $\Delta pT$.
By the choice of the schedule, the plain average payoff of the scheduled $T$-block is within $\xi$ of $\hat r_j$.
If a false punishment occurs, the total reward over the punishment block differs from $K\hat r_j$ by at most $2U^\star K$.
Therefore
\[
\left|\mathbb E[R_{j,l}]-\hat r_j\,\mathbb E[\ell_l]\right|
\le
\Delta pT+\xi T+2U^\star e^{-2p^2T}K.
\]
Since $\mathbb E[\ell_l]\ge T$ and $K/T\le c+1$, we obtain
\[
\left|
\frac{\mathbb E[R_{j,l}]}{\mathbb E[\ell_l]}-\hat r_j
\right|
\le
\Delta p+\xi+2U^\star e^{-2p^2T}(c+1)
\le
3\xi,
\]
by the parameter choice above.
\end{proof}

\begin{lemma}[On-path discounted payoffs]
\label{lem:on-path-lean}
There exists a cutoff below $1$ such that for every $\delta$ above that cutoff, every phase index $l$, and every LLM $j$,
\[
\left|
U_j^\delta(S^{p,T,K}\mid \text{start in phase }l)-\hat r_j
\right|
\le
4\xi.
\]
In particular,
\[
\left|U_j^\delta(S^{p,T,K})-r_j\right|\le \gamma.
\]
\end{lemma}

\begin{proof}
Let
\[
W_{j,l}^\delta
:=
\mathbb E\!\left[\sum_{t\ge 1}\delta^{t-1}U_j(\bm M_t)\,\middle|\,\text{start in phase }l\right].
\]
Because honest cycles are i.i.d., if $G_{j,l}^\delta$ denotes the discounted reward within the first cycle, then
\[
W_{j,l}^\delta
=
G_{j,l}^\delta+\mathbb E[\delta^{\ell_l}]\,W_{j,l}^\delta.
\]
Hence
\[
(1-\delta)W_{j,l}^\delta
=
(1-\delta)\frac{G_{j,l}^\delta}{1-\mathbb E[\delta^{\ell_l}]}.
\]
Since $\ell_l\le T+K$ almost surely, we have
\[
G_{j,l}^\delta\to \mathbb E[R_{j,l}]
\qquad\text{and}\qquad
\frac{1-\mathbb E[\delta^{\ell_l}]}{1-\delta}\to \mathbb E[\ell_l]
\qquad
(\delta\uparrow 1),
\]
so
\[
\lim_{\delta\uparrow 1}(1-\delta)W_{j,l}^\delta=\frac{\mathbb E[R_{j,l}]}{\mathbb E[\ell_l]}.
\]
Since the set of pairs $(j,l)\in L\times L$ is finite, there exists a cutoff below $1$ such that for every phase index $l$, every LLM $j$, and every $\delta$ above that cutoff,
\[
\left|
U_j^\delta(S^{p,T,K}\mid \text{start in phase }l)-\frac{\mathbb E[R_{j,l}]}{\mathbb E[\ell_l]}
\right|
\le
\xi
\]
Combining this with Lemma~\ref{lem:on-path-cycle-lean} yields, for every such $\delta$,
\[
\left|
U_j^\delta(S^{p,T,K}\mid \text{start in phase }l)-\hat r_j
\right|
\le
4\xi.
\]
Since the initial public state is the start of a fresh phase-$1$ block, we also have
\[
\left|U_j^\delta(S^{p,T,K})-\hat r_j\right|\le 4\xi.
\]
Using $|\hat r_j-r_j|\le \xi$ yields
\[
\left|U_j^\delta(S^{p,T,K})-r_j\right|
\le
\left|U_j^\delta(S^{p,T,K})-\hat r_j\right|
+
\left|\hat r_j-r_j\right|
\le
5\xi
<
\gamma.
\]
\end{proof}

\begin{lemma}[Changing few time periods in a fixed block]
\label{lem:theta-lean}
Fix integers $T\ge 1$ and $q\in\{0,\dots,T\}$, and a tolerance $\zeta>0$.
Then there exists a cutoff below $1$ such that if two payoff sequences $x_1,\dots,x_T$ and $y_1,\dots,y_T$ lie in $[-U^\star,U^\star]$, satisfy $|x_s-y_s|\le \Delta$ for all $s$, and differ on at most $q$ indices, then for every $\delta$ above that cutoff,
\[
\left|
\frac{1-\delta}{1-\delta^T}\sum_{s=1}^T \delta^{s-1}(x_s-y_s)
\right|
\le
\Delta\frac{q}{T}+\zeta.
\]
\end{lemma}

\begin{proof}
Let $I:=\{s\le T:x_s\neq y_s\}$. Then $|I|\le q$, so
\[
\left|
\frac{1-\delta}{1-\delta^T}\sum_{s=1}^T \delta^{s-1}(x_s-y_s)
\right|
\le
\Delta\frac{(1-\delta)\sum_{s\in I}\delta^{s-1}}{1-\delta^T}.
\]
Because $\delta^{s-1}$ decreases in $s$, the right-hand side is maximized when $I=\{1,\dots,q\}$. Hence
\[
\Delta\frac{(1-\delta)\sum_{s\in I}\delta^{s-1}}{1-\delta^T}
\le
\Delta\frac{1-\delta^q}{1-\delta^T}.
\]
As $\delta\uparrow 1$, the last expression converges to $\Delta q/T$, so it is at most $\Delta q/T+\zeta$ for all $\delta$ sufficiently close to $1$.
\end{proof}

\begin{lemma}[Deviation gain in pre-$o$ phases]
\label{lem:deviation-pre-lean}
There exists a cutoff below $1$ such that for every $\delta$ above that cutoff, every LLM $o\in L$, and every unilateral deviation $S'_o$, the contribution to
\[
U_o^\delta(S'_o,S_{-o})-U_o^\delta(S_o,S_{-o})
\]
coming from pre-$o$ phases and their immediate punishment blocks is at most $4\xi$.
\end{lemma}

\begin{proof}
Fix a deviator $o\in L$ and a unilateral deviation $S'_o$. Relabel the phases cyclically so that the current search order is $1,2,\ldots,o$.

Consider a phase-$l$ block with $l<o$.
Let $X_B$ be the number of time periods in that block on which $o$'s meta-action differs from the prescribed one on some role $i$ with $p_i^o>0$.
Let $P_B$ be the number of probe time periods of the honest prober $l$.
Call the block:
\begin{itemize}
    \item \emph{light} if $X_B\le pT$;
    \item \emph{heavy} if $X_B>pT$;
    \item \emph{probe-good} if $P_B\le 2pT$;
    \item \emph{probe-bad} otherwise.
\end{itemize}
Again,
\[
\Pr(\text{probe-bad}\mid \text{history before the block})\le e^{-2p^2T}.
\]

If a pre-$o$ block is light and probe-good, then
\[
d^l\le \frac{P_B+X_B}{T}\le 3p=\tau,
\]
so punishment is not triggered.
Relative to the honest path, the deviator can change payoffs only on at most $pT$ time periods, each by at most $\Delta$.
Applying Lemma~\ref{lem:theta-lean} with $q=\lceil pT\rceil$ and tolerance $\xi/2$, the discounted gain from such a block is at most the block's own normalized weight times
\[
\Delta\frac{\lceil pT\rceil}{T}+\frac{\xi}{2}\le 2\xi
\]
for all $\delta$ sufficiently close to $1$, by \eqref{eq:large-T-lean}.

Now consider a heavy pre-$o$ block. On each deviating time period there exists some role-action pair $(i,a)$ such that
\[
M_i^{\prime o}(a)>\overline M_i^{h,o}(a)
\]
for the currently scheduled block $h$.
If on that time period the honest prober $l$ probes and draws a meta-action that plays action $a$ in role $i$, then by \eqref{eq:tot-linear-lean},
\[
H^l_{h,s}(i,a)>\Gamma_{h,l}(i,a),
\]
so rule (1) advances the phase.
The probability of this event on a given deviating time period is at least $p/|A_i|$. Since $R$ and the action sets are finite, this probability admits a strictly positive lower bound that is uniform over roles.
Because a heavy block contains at least $\lceil pT\rceil$ deviating time periods and the probing draws are independent across time periods, there exists a constant $\rho>0$ such that every heavy pre-$o$ block advances the phase with probability at least $\rho$.
It follows that the number of heavy pre-$o$ blocks before the process reaches phase $o$ is stochastically dominated by a sum of at most $k-1$ geometric random variables with mean $1/\rho$.
If $N_{\mathrm{heavy}}$ is the total number of time periods contained in heavy pre-$o$ blocks together with their immediate punishment blocks, then
\[
\mathbb E[N_{\mathrm{heavy}}]\le \frac{(k-1)(T+K)}{\rho}.
\]
Hence their total discounted contribution is at most
\[
2U^\star(1-\delta)\,\frac{(k-1)(T+K)}{\rho},
\]
which tends to $0$ as $\delta\uparrow 1$.

Finally, consider probe-bad pre-$o$ blocks.
Associate to each such block the chunk consisting of the $T$-block itself and the immediate punishment block if one occurs.
Its length is at most $T+K$.
If $w(B)$ is the normalized discount weight of the $T$ time periods of the block, then the weight of the associated chunk is at most
\[
\frac{1-\delta^{T+K}}{1-\delta^T}\,w(B).
\]
Since
\[
\frac{1-\delta^{T+K}}{1-\delta^T}\longrightarrow \frac{T+K}{T}\le 2+c
\qquad
(\delta\uparrow 1),
\]
there exists a cutoff below $1$ such that for every $\delta$ above that cutoff,
\[
\frac{1-\delta^{T+K}}{1-\delta^T}\le 2+c,
\qquad\text{and}\qquad
2U^\star(1-\delta)\,\frac{(k-1)(T+K)}{\rho}\le \xi.
\]
After taking that cutoff closer to $1$ if necessary, we may also assume that the light-block estimate above holds for every $\delta$ above it.
For every such $\delta$, the total normalized weight of the disjoint pre-$o$ blocks is at most $1$, so the expected contribution of all probe-bad pre-$o$ chunks is at most
\[
2U^\star e^{-2p^2T}(2+c)\le \xi.
\]
The light, probe-good pre-$o$ blocks are also disjoint, so their normalized discounted weights sum to at most $1$; summing the per-block bound above therefore gives a total contribution of at most $2\xi$.

For every such $\delta$:
\begin{itemize}
    \item light, probe-good pre-$o$ blocks contribute at most $2\xi$ in total;
    \item heavy pre-$o$ chunks contribute at most $\xi$ in expectation;
    \item probe-bad pre-$o$ chunks contribute at most $\xi$ in expectation.
\end{itemize}
Hence the total contribution of all pre-$o$ phases is at most
\[
2\xi+\xi+\xi=4\xi.
\]
\end{proof}

\begin{lemma}[Deviation gain in phase $o$]
\label{lem:deviation-own-lean}
There exists a cutoff below $1$ such that for every $\delta$ above that cutoff, every LLM $o\in L$, and every unilateral deviation $S'_o$, the contribution to
\[
U_o^\delta(S'_o,S_{-o})-U_o^\delta(S_o,S_{-o})
\]
coming from phase-$o$ blocks and their punishment cycles is at most $8\xi$.
\end{lemma}

\begin{proof}
Fix a deviator $o\in L$ and a unilateral deviation $S'_o$. Relabel the phases cyclically so that the current search order is $1,2,\ldots,o$.
By Lemma~\ref{lem:on-path-lean}, at the start of any fresh phase-$l$ block the honest continuation payoff of $o$ satisfies
\begin{equation}
U_o^\delta(S^{p,T,K}\mid \text{start in phase }l)\ge \hat r_o-4\xi
\label{eq:honest-lower-lean}
\end{equation}
for all $\delta$ sufficiently close to $1$.

Call a phase-$o$ block \emph{light} if its discrepancy fraction is at most $\tau$, and \emph{heavy} otherwise.

If a phase-$o$ block is light, then at most $\tau T$ time periods differ from the intended public path.
By \eqref{eq:tot-linear-lean}, once the other LLMs are fixed, matching the intended public path is equivalent to matching the prescribed role-wise meta-action of $o$ on all relevant roles.
Applying Lemma~\ref{lem:theta-lean} with $q=\lceil \tau T\rceil$ and tolerance $\xi/2$, the discounted gain created by changing those time periods is at most the block's own normalized weight times
\[
\Delta\frac{\lceil \tau T\rceil}{T}+\frac{\xi}{2}\le 2\xi
\]
for all $\delta$ sufficiently close to $1$, by \eqref{eq:large-T-lean}.
By the choice of the schedule, the plain average payoff of the intended $T$-time-period block differs from $\hat r_o$ by at most $\xi$.
Applying Lemma~\ref{lem:discount-cycle-lean} with $S=T$ and tolerance $\xi$, the normalized discounted average of that intended block differs from $\hat r_o$ by at most $2\xi$ for all $\delta$ sufficiently close to $1$.
Combining this with \eqref{eq:honest-lower-lean}, the total gain from a light phase-$o$ block, relative to the honest continuation benchmark, is at most its own weight times
\[
2\xi+2\xi+4\xi=8\xi.
\]

If a phase-$o$ block is heavy, then rule (2) triggers punishment.
Since $x\mapsto (TU_o^{\max}+x\,\mathrm{IR}_o)/(T+x)$ is decreasing and $K\ge cT$, the plain average payoff of $o$ over the resulting $(T+K)$-cycle is at most
\[
\frac{TU_o^{\max}+K\,\mathrm{IR}_o}{T+K}
\le
\frac{U_o^{\max}+c\,\mathrm{IR}_o}{1+c}
\le
\hat r_o.
\]
Applying Lemma~\ref{lem:discount-cycle-lean} with $S=T+K$ and tolerance $\xi$, the discounted gain of a heavy phase-$o$ cycle, relative to the honest continuation benchmark, is at most its cycle-weight times
\[
\xi+4\xi=5\xi
\]
for all $\delta$ sufficiently close to $1$.

In phase $o$, rule (1) cannot fire because by the definition of $\Gamma_{h,o}$ and \eqref{eq:tot-linear-lean} one always has
\[
H^o_{h,s}(i,a)\le \Gamma_{h,o}(i,a)
\qquad
(\forall h,s,i,a).
\]
Hence every phase-$o$ time period belongs either to a light phase-$o$ block or to a heavy phase-$o$ cycle, and the normalized weights of these pieces sum to at most $1$.
Taking $\delta$ sufficiently close to $1$ so that all of the preceding bounds hold simultaneously, we obtain:
\begin{itemize}
    \item light phase-$o$ blocks contribute at most $8\xi$ in total;
    \item heavy phase-$o$ cycles contribute at most $5\xi$ in total.
\end{itemize}
Therefore the total contribution of all phase-$o$ blocks and their punishment cycles is at most
\[
\max\{8\xi,5\xi\}=8\xi.
\]
\end{proof}

\begin{lemma}[Any unilateral deviation gains at most $\epsilon$]
\label{lem:deviation-lean}
There exists $\bar\delta<1$ such that for every $\delta\in(\bar\delta,1)$, every LLM $o\in L$, and every unilateral deviation $S'_o$,
\[
U_o^\delta(S'_o,S_{-o})-U_o^\delta(S_o,S_{-o})\le \epsilon.
\]
\end{lemma}

\begin{proof}
Choose $\bar\delta<1$ so large that the conclusions of Lemmas~\ref{lem:on-path-lean}, \ref{lem:deviation-pre-lean}, and \ref{lem:deviation-own-lean} all hold for every $\delta\in(\bar\delta,1)$.
Fix $\delta\in(\bar\delta,1)$, an LLM $o\in L$, and a unilateral deviation $S'_o$.
Then Lemmas~\ref{lem:deviation-pre-lean} and \ref{lem:deviation-own-lean} imply that
\[
U_o^\delta(S'_o,S_{-o})-U_o^\delta(S_o,S_{-o})
\le
4\xi+8\xi
=
12\xi
<
\epsilon.
\]
\end{proof}

\myparagraph{Conclusion.}
Take $\bar\delta$ from Lemma~\ref{lem:deviation-lean}. By its choice, Lemma~\ref{lem:on-path-lean} gives the payoff approximation and Lemma~\ref{lem:deviation-lean} gives incentive compatibility.
Therefore, for every $\delta\in(\bar\delta,1)$, the strategy profile $S^{p,T,K}$ is an $\epsilon$-equilibrium and its normalized discounted payoff vector is within $\gamma$ of $r$.
\label{par:endOfProof}


\end{document}